\documentclass[11pt]{article}

\usepackage{amssymb,amsthm,amsmath}
\usepackage{algorithm}
\usepackage{algpseudocode}


\usepackage{subfig}
\usepackage[usenames,dvipsnames]{xcolor}
\usepackage[colorlinks,citecolor=blue,linkcolor=BrickRed]{hyperref}\usepackage[colorlinks,citecolor=blue,linkcolor=BrickRed]{hyperref}
\usepackage{fullpage}
\usepackage{graphicx}
\usepackage{enumitem}
\usepackage{thmtools,thm-restate}
\usepackage{wrapfig}
\usepackage{comment}
\usepackage{cleveref}

\graphicspath{{figures/}}

\newtheorem{theorem}{Theorem}[section]

\newtheorem{lemma}[theorem]{Lemma}

\newtheorem{definition}[theorem]{Definition}

\newtheorem{remark}[theorem]{Remark}

\newif\ifFULL
\FULLtrue

\newcommand{\R}{\mathbb{R}}

\newcommand{\eps}{\varepsilon}
\newcommand{\Omgt}{\Tilde{\Omega}}
\newcommand{\Ot}{\Tilde{O}}

\newcommand{\poly}{\mathsf{poly}}
\newcommand{\polylog}{\mathsf{polylog}}

\newcommand{\ind}{\mathbf{1}}
\newcommand{\LHS}{\mathsf{LHS}}
\newcommand{\RHS}{\mathsf{RHS}}

\newcommand{\bone}{\mathbf{1}}

\newcommand{\ALG}{\mathsf{ALG}}
\newcommand{\EO}{\mathsf{EO}}

\allowdisplaybreaks

\newenvironment{proofof}[1]{\smallskip\noindent{\bf Proof of #1.}}%
        {\hspace*{\fill}$\Box$\par}

\title{Improved Lower Bounds for Submodular Function Minimization}

\ifdefined\isFocsSubmission
\author{Anonymous Authors}
\else
\author{
Deeparnab Chakrabarty 
\thanks{Dartmouth College, \texttt{deeparnab@dartmouth.edu}.}
\and 
Andrei Graur \thanks{Stanford University, \texttt{agraur@stanford.edu}.}
\and 
Haotian Jiang \thanks{University of Washington, \texttt{jhtdavid@cs.washington.edu}.}
\and 
Aaron Sidford \thanks{Stanford University, \texttt{ sidford@stanford.edu}.}
}
\fi

\date{\today}

\begin{document}

\begin{titlepage}
  \maketitle
  
\begin{abstract}

We provide a generic technique for constructing families of submodular functions to obtain lower bounds for submodular function minimization (SFM). Applying this technique, we prove that any deterministic SFM algorithm on a ground set of $n$ elements requires at least $\Omega(n \log n)$ queries to an evaluation oracle. This is the first super-linear query complexity lower bound for SFM and improves upon the previous best lower bound of $2n$ given by [Graur et al., ITCS 2020]. Using our construction, we also prove that any (possibly randomized) parallel SFM algorithm, which can make up to $\mathsf{poly}(n)$ queries per round, requires at least $\Omega(n / \log n)$ rounds to minimize a submodular function. This improves upon the previous best lower bound of $\tilde{\Omega}(n^{1/3})$ rounds due to [Chakrabarty et al., FOCS 2021], and settles the parallel complexity of query-efficient SFM up to logarithmic factors due to a recent advance in [Jiang, SODA 2021].

  \end{abstract}
  \thispagestyle{empty}
\end{titlepage}

%%%%%%%%%%%%%%%%%%%%%%%%%%%%%%%%%
% Sec: Introduction
%%%%%%%%%%%%%%%%%%%%%%%%%%%%%%%%%

\section{Introduction}

A real-valued function $f : 2^V \rightarrow \mathbb{R}$ defined on subsets of an $n$-element ground set $V$ is \emph{submodular} if $f(X \cup \{e\}) - f(X) \geq f(Y\cup \{e\}) - f(Y)$ for any $X \subseteq Y \subseteq V$ and $e\in V\setminus Y$.
Submodular functions are ubiquitous and include cut functions in (hyper-)graphs, set coverage functions, rank functions of matroids, utility functions in economics, and entropy functions in information theory, etc.

Given the expressive power of submodular functions, the  optimization of them has been extensively studied. The problem of submodular function minimization (SFM), i.e. $\min_{S \subseteq V} f(S)$, given black-box access to an \textit{evaluation oracle}, which returns the value $f(S)$ upon receiving a set $S \subseteq V$, 
encompasses many important problems in theoretical computer science, operations research, game theory, and more. Recently, SFM has found applications in computer vision, machine learning, and speech recognition~\cite{BVZ01, KKT08, KT10, LB11}. Correspondingly, SFM has been the subject of extensive research for decades and is foundational to the theory of combinatorial optimization.

Throughout the paper, unless specified otherwise, we focus on the {\em strongly-polynomial} regime for the {\em query complexity} of SFM. 
We refer to an SFM algorithm as strongly-polynomial (in terms of query complexity) if the number of evaluation oracle queries it makes is at most a polynomial in $n$ and does not depend on the range of the function. 
After decades of advances \cite{GLS81, C85, GLS88, S00, FI00, IFF01, I03, V03, O09, IO09}, the current state-of-the-art strongly-polynomial algorithms include an $O(n^2 \log n)$-query, $\exp(O(n))$-time algorithm \cite{J21} and an $O(n^3 \log \log n/\log n)$-query, $\poly(n)$-time algorithm \cite{J21}, which improved (in query complexity) upon $\tilde{O}(n^3)$-query, $\tilde{O}(n^4)$-time algorithms of \cite{LSW15, JLSW20, DVZ21}.\footnote{Throughout, we use $\tilde{O}(\cdot)$ to hide polylogarithmic factors.}

Despite the rich history of SFM research, obtaining {\em lower bounds} on the query complexity for SFM has been notoriously difficult. \cite{H08} described two different constructions of submodular functions whose minimization requires $n$-queries to an evaluation oracle; in fact, both can be minimized by querying all the $n$ singletons. Later,~\cite{CLSW17} showed that one of the examples in ~\cite{H08} also needs $n/4$ gradient queries to the Lov\'asz extension of the submodular function. This remained the best lower bound, until recently~\cite{GPRW19} proved a $2n$-query lower bound 
on SFM via a non-trivial construction of a submodular function (which can be minimized in $2n$ queries). For more discussions on difficulties in obtaining super-linear lower bounds, we refer the reader to~\Cref{subsec:related-work}.

More recently, there has been an interest in understanding the {\em parallel complexity} of SFM. Note that any SFM algorithm proceeds by making queries to an evaluation oracle in rounds, and the parallel complexity of  SFM is the minimum number of rounds (also known as the depth) required by any {\em query-efficient} SFM algorithm that makes at most $\poly(n)$ evaluation oracle queries. 
All SFM algorithms described above proceed in $\Omega(n)$-rounds. The best known round-complexity is the algorithm due to~\cite{J21} which runs in $O(n\log n)$ rounds.
On the lower bound side,  
\cite{BS20} proved that any query-efficient SFM algorithm must proceed in $\Omega(\log n/\log\log n)$-rounds. This was improved in~\cite{CCK21} to an $\Omgt(n^{1/3})$-lower bound on the number of rounds for query-efficient SFM. The latter paper also mentioned a bottleneck of $n^{1/3}$ to their approach and left open the question of whether a nearly-linear number of rounds are needed, or whether there is a query-efficient SFM algorithm proceeding in  $n^{1-\delta}$ many rounds for some absolute constant $\delta > 0$.

%%%%%%%%%%%%%%%%%%%%%%%%%%%%%%%%%
% Subsec: Our Results
%%%%%%%%%%%%%%%%%%%%%%%%%%%%%%%%%

\subsection{Our Results.} 

In this paper we provide improved lower bounds for both the query complexity for SFM, and the round complexity for query-efficient parallel SFM. We prove that any deterministic SFM algorithm requires $\Omega(n \log n)$ queries to an evaluation oracle, and that any parallel SFM algorithm making at most $\poly(n)$ queries must proceed in $\Omega(n / \log n)$ rounds.

\begin{restatable}[Query complexity lower bound for deterministic algorithms]{theorem}{OracleLowerBound} \label{thm:oracle_complexity_lb_det}
For any finite set $V$ with $n$ elements and deterministic SFM algorithm $\ALG$, there exists a submodular function $F: 2^V \rightarrow \mathbb{R}$ such that $\ALG$ makes at least $\frac{n}{2} \log_2 (\frac{n}{4})$ evaluation oracle queries to minimize $F$. 
\end{restatable}

\Cref{thm:oracle_complexity_lb_det} constitutes the first super-linear lower bound on the number of evaluation queries for SFM. The previous best lower bound was $2n$, due to \cite{GPRW19}.

\begin{restatable}[Parallel lower bound for randomized algorithms]{theorem}{ParallelLowerBound} \label{thm:adaptive_lb}
For any finite set $V$ with $n$ elements, constant $C \geq 2$, and (possibly randomized) parallel SFM algorithm $\ALG$ that makes at most $Q := n^C$ queries per round, there exists a submodular function $F: 2^V \rightarrow \mathbb{R}$ such that $\ALG$ takes at least $\frac{n}{2C \log_2 n}$ rounds to minimize $F$ with high probability.
\end{restatable}

\Cref{thm:adaptive_lb} improves upon the previous best $\Omgt(n^{1/3})$ parallel lower bound due to~\cite{CCK21}. Further, \Cref{thm:adaptive_lb} 
is optimal up to logarithmic factors due to \cite{J21}, which
yields an $O(n \log n)$-round, $O(\poly(n))$-queries algorithm.\footnote{This query bound is due to the fact that an algorithm in \cite{J21} solves SFM with $O(n \log n)$ computations of the 
subgradients of the Lov\'asz extension. Further, each computation of a subgradient can be implemented by making $n$ queries to an evaluation oracle for the submodular function in parallel, i.e.\ a single round.}. 

Both \Cref{thm:oracle_complexity_lb_det} and \Cref{thm:adaptive_lb} are obtained by constructing a new family of submodular functions. This family of submodular functions and the analysis of their properties is our main technical contribution. At a high level, 
we glue together simple submodular functions, each of which is defined on a distinct part of a large partition of the ground set $V$ and has a unique minimizer. The main novelty of our construction is an approach to assemble these functions into a layered structure in such a way that any SFM algorithm needs to effectively find the minimizer of one layer before obtaining any information about the functions in later layers. 
This forces any parallel algorithm to have depth equal to the number of parts, which implies our parallel lower bound. We also show that minimizing a single part needs a number of queries super-linear in the size of that part,  implying the super-linear query complexity lower bound for deterministic algorithms. 
More insights into our construction and proofs are given in \Cref{subsec:techniques}.

%%%%%%%%%%%%%%%%%%%%%%%%%%%%%%%%%
% Subsec: Our Techniques
%%%%%%%%%%%%%%%%%%%%%%%%%%%%%%%%%

\subsection{Our Techniques}
\label{subsec:techniques}

Previous works on proving lower bounds for parallel SFM \cite{BS20, CCK21} apply the following
generic framework. At a high level, they design a family of hard submodular functions which are parameterized using a partition $(P_1, \ldots, P_\ell)$ of the ground set.
The key property they show is that even after obtaining answers to polynomially many queries in round $i$, any algorithm (with high probability) doesn't
possess any information about the elements in $P_{i+1}, \ldots, P_\ell$.
Further, the construction also has the property that knowing which elements are in the final part $P_\ell$ is crucial in obtaining the minimizer.
These properties prove an $\ell-1$ lower bound on the number of rounds for parallel SFM.

Our paper also proceeds under the same generic framework, but departs crucially from prior work in the design of the family of hard submodular functions $\mathcal{F}$, which is the main technical innovation of this paper. 
With this new construction, our query complexity lower bound follows by a careful adversarial choice of function $F \in \mathcal{F}$, and our parallel round complexity lower bound follows by choosing a random function uniformly at random from $\mathcal{F}$. 

\medskip
\noindent \textbf{Recap of Previous Constructions.}
Before we dive into a high-level discussion of our construction, here we remind the reader
of the construction ideas in \cite{BS20} and \cite{CCK21}, and why they stop short of proving a nearly-linear lower bound on the number of rounds for parallel SFM.
Both these works construct so-called {\em partition submodular 
functions} $F$ where one is given a partition $(P_1, \ldots, P_\ell)$, 
and the value of $F(S)$ depends only on the {\em cardinality} of the sets $|S\cap P_1|, \ldots, |S\cap P_\ell|$. 
Note that when the algorithm has no information about $P_1, \ldots, P_\ell$, for instance in the first round of querying, then for any query set $S$, 
these cardinalities are roughly proportional to the cardinalities of each part.
The main idea behind the constructions in \cite{CCK21, BS20} is to come up with 
submodular functions where this ``roughly proportional'' property is used to hide any information about the parts $P_2,\ldots, P_\ell$.
However, the fact that $|S\cap P_i|$'s can typically differ by a standard deviation necessarily requires each part $P_i$ to be ``sufficiently large'' and this, in turn,
puts a $o(n)$ bottleneck on the {\em number} of parts $\ell$.
As it stands, it is not clear how to obtain a better than $n^{1/3}$-lower bound on the round complexity of parallel SFM
using partition submodular functions.

Interestingly, a similar approach as above has also been the main tool to prove
lower bounds for parallel convex optimization~\cite{N94,BS18,BJL+19,DG19}.
We defer to \Cref{subsec:related-work} for a more detailed discussion of this broader context.

\medskip \noindent \textbf{Ideas Behind our Construction.} Our construction deviates from the notion of partition submodular functions in that the function value $F(S)$ crucially depends on the {\em identity} of the set $S\cap P_i$ rather than the size, which helps us bypass the bottleneck in previous constructions and obtain nearly-linear lower bound on the number of rounds. 

It is convenient to think of the family of functions we construct
in a recursive fashion.
Pick a subset $A\subseteq V$ of size $2r$, which corresponds to the first part $P_1$ in the partition described above, and denote $B:= V \setminus A$ the remainder parts $P_2 \cup \cdots \cup P_\ell$.  
For notational convenience, we denote $S_A := S \cap A$ and $S_B := S \cap B$ for any set $S \subseteq V$. 
Let $R \subseteq A$ be a subset of size $|R| = r = |A|/2$, and consider the following function $F: 2^V \to \R$ defined as
\begin{equation} 
	F(S) := h_R(S) + \beta \cdot \ind(S_A = R) \cdot g(S_B) \label{eq:meta-def}, \tag{Meta Definition}
\end{equation}
where $\ind(\cdot)$ is the indicator function, and $g$ is a submodular function which will recursively be the same as $F$ defined over the smaller universe $B$. 
The parameter $\beta$ is a small scalar, and should be thought of as $\Theta(\frac{1}{|V|})$.
We aim to design the function $h_R(\cdot)$ to have the following two properties:
\begin{itemize}
\vspace{-0.2cm}
	\item[(P1)] Any set $S\subseteq V$ is a minimizer of $h_R$ if and only if $S_A = R$, 
\vspace{-0.1cm}
	\item[(P2)] The function $F$ defined in \eqref{eq:meta-def} is submodular whenever $g$ is submodular.
\vspace{-0.2cm}
\end{itemize}

We now claim that obtaining such a function $h_R$ suffices to prove an $\frac{n}{2C\log n}$-lower bound on the number of rounds required by any exact parallel SFM algorithm making $\leq n^C$ queries per round. In particular, the subsets $R \subseteq A \subseteq V$ with $|R| = |A|/2 = C \log n$, as well as the recursively defined function $g$, will be chosen uniformly at random. 

To see this, first observe that when $\beta$ is sufficiently small, if $S_g^*$ is a (unique) minimizer of the function $g$, then the set $S^*:= R \cup S_g^*$ is a (unique) minimizer of $F$. This crucially uses property (P1) which says that $R \cup S_B$ is a minimizer of $h_R$ for any $S_B \subseteq B$.
Next, consider the first round of queries $Q^1, \ldots, Q^T$. Since $R \subseteq A$ is chosen uniformly at random, and because $|R| = |A|/2 = C \log n$, the probability that one of these $Q^i_A = R$ is negligible if $T \leq n^C$. 
Therefore, all the answers to the queries in the first round are precisely $h_R(Q_i)$, revealing no information about the function $g$.
On the other hand, the minimizer of $F$ needs to minimize $g$.
Therefore, if we pick $g$ randomly from the same family of $F$ but over the smaller universe $B$, we could apply the above argument recursively with $2C\log n$ fewer elements and one fewer round.
In this way, we prove an $\frac{n}{2C\log n}$-lower bound on the number of rounds needed to exactly minimize the random submodular function $F$.

The big question left, of course, is whether one can construct a function $h_R$ with the properties mentioned above. This is what we discuss next. 

\medskip \noindent \textbf{Obtaining Submodularity.} Let us first discuss an idea which does not work, and then fix it.
One way to define $h_R$ is to take a submodular function $f_R$ defined {\em only} over elements of $A$, whose (unique) minimizer is the subset $R$, and then extend it 
as $h_R(S) := f_R(S_A)$. In particular, 
\begin{equation}\label{eq:try1}
F(S) := f_R(S_ A) + \beta \cdot \ind(S_A = R) \cdot g(S_B) . \tag{First Try}
\end{equation}
\noindent
Note that it satisfies property (P1), i.e. $S$ is a minimizer of $h_R$ if and only if $S_A = R$. 
Unfortunately, 
the resulting function $F$ may not be submodular even if both $f_R$ and $g$ are submodular. 
To see this, consider an element $e\in B$ and consider the marginal increase in $F$ when $e$ is added to a set $S$. 
Since $f_R$ only depends on $S_A$ and $e\in B$, in the marginal calculation of $F(S + e) - F(S)$, the $f_R$ terms cancel out. In particular, we get that 
\[
F(S + e) - F(S) = \beta\cdot \ind(S_A = R) \cdot \left(g(S_B + e) - g(S_B)\right) .
\]
Suppose the parenthesized term is positive for some $S_B$ (e.g. the maximal minimizer of $g$) and 
consider the sets $S := R\cup S_B$ and $S' := R' \cup S_B$, where $R'$ is any strict subset of $R$. In this case $F(S + e) - F(S) > 0$ while $ F(S'+ e) - F(S') = 0$ and
since $S' \subseteq S$, this violates submodularity. \smallskip

To fix the above idea, we pad the function $f_R(S_A)$ with what we call a ``{\em submodularizer function}'' $\phi(S)$. Think of $\phi$ as taking two sets $(S_A, S_B)$ as input;
the first set is a subset of $A$ the other is a subset of $B$. We define $h_R(S) := f_R(S_A) + \phi(S_A, S_ B)$ and therefore, 
\begin{equation}\label{eq:try2}
F(S) := f_R(S_A) + \phi(S_A, S_B) + \beta \cdot \ind(S_A = R) \cdot g(S_B) . \tag{Layered Function}
\end{equation}
What properties do we need from $\phi$? First, since (P1) requires that when $S_A = R$, the set $S$ is a minimizer of $f+\phi$ irrespective 
of what $S_B$ is, this suggests $\phi(R, S_B)$ is the same for any $S_B \subseteq B$. For simplicity, assume this is $0$.
That is, when $S_A = R$, the $\phi$ function doesn't have any effect.
However, considering the reason our first attempt failed, when $S'_A$ is a {\em strict} subset of $R$, then $\phi(S'_A, S_B)$ should be so defined such that adding an element $e \in B$ to $S_B$ {\em strictly increases} the function value. 
This would make sure that $F(S' + e) - F(S') > 0$ for the violating example in the previous paragraph. Not only that, this strict increase should be {\em greater}
than the increase in $F(S + e) - F(S)$, where $S = (R,S_B)$ is as in the previous paragraph, and this increase is $\beta$ times some marginal of $g$. 
To ensure that this occurs, we choose $\beta$ to be ``small enough''; it suffices to choose a constant factor less than the strict
increase of the function $\phi$. A similar argument also leads us to the conclusion that when $S_A$ is a {\em strict superset} of $R$, then $\phi(S_A, S_B)$ should {\em strictly decrease} in value when an element is added to $S_B$. A definition of $\phi$ that works is the following:
\begin{equation}\label{eq:def-phi}\tag{Submodularizer}
\phi(S_A, S_B) := \begin{cases}
+4\beta |S_B| & \textrm{if $S_A$ strict subset of $R$} \\
-4\beta |S_B| & \textrm{if $S_A$ strict superset of $R$} \\
0 & \textrm{otherwise, and in particular if $S_A = R$}
\end{cases} 
\end{equation}
Note we still have the parameter $\beta$ unspecified, and we set it soon. 

The above discussion only considered marginals of an element $e\in B$ to the function $F$. One also needs to be careful about the case when the element $e \in A$.
This will put a restriction on what $f_R$ and $\beta$ are, and will form the last part of our informal description. 

Consider an element $e\in A\setminus R$ and consider the function $\phi(R, S_B)$ for an arbitrary $S_B \subseteq B$. Note that, as defined, the value of $\phi(R, S_B) = 0$
and $\phi(R+e, S_B) = - 4\beta |S_B|$. That is, adding $e$ to $R\cup S_B$ can {\em decrease} the $\phi$ function value by $-4\beta |S_B|$. On the other hand, adding $e$ to $(A-e)\cup S_B$ 
doesn't change the $\phi$-value. Indeed, $\phi(A, S_B) = \phi(A - e, S_B) = -4\beta |S_B|$ since both $A$ and $A-e$ are strict supersets of $R$ (remember $e\notin R$). 
In short, the function $\phi$ is {\em not} submodular and this endangers the submodularity of  the sum function $h_R = f_R + \phi$.

To fix this, we {\em make sure} that the function $f_R$ has a ``large gap'' between $f_R(R + e)$ and $f_R(R)$. In particular, we ensure that $f_R(R+e) - f_R(R) = \Omega(1)$
while $\beta = O(1/n)$. In this way, although adding $e \in A \setminus R$ to $(R, S_B)$ can decrease the $\phi$ value by $-4 \beta |S_B|$, since $\beta = O(1/n)$ this decrease is smaller than the increase caused by $f_R(R+e) - f_R(R)$ when the constants are properly chosen. In particular, we define the function $f_R$ on the universe $A$ as follows
\begin{equation}\label{eq:def-fR}
f_R(S_A) := \begin{cases}
0 & \text{if $S_A = R$}\\
1 & \text{if $S_A$ is a strict superset or a strict subset of $R$}\\
2 & \text{otherwise} \\
\end{cases}
\end{equation}
It is not too hard to see that this function $f_R$ is submodular; in fact, this function (or a scaled version if it) has been considered before in the submodular function literature~\cite{H08, CLSW17}.
This completes the informal description and motivation of our construction of hard functions;
a formal presentation of our construction and the full proof of its properties can be found in \Cref{sec:construction} and \Cref{sec:missing_proofs}.

\medskip \noindent \textbf{Query Complexity Lower Bound.} 
While discussed and motivated in terms of the number of parallel rounds for SFM, our construction can also prove an $\Omega(n\log n)$ lower bound on the {\em query complexity} of any {\em deterministic} SFM algorithm.
Indeed, for this part, we consider the family where the size of $|A| = 2$, and $R$ is a singleton among these two elements.
Instead of selecting a random function from this family, we adversarially choose a worst-case function depending on the deterministic algorithm. 
Note that the function definition above doesn't require the size $|A|$ to be large; we made it large in the previous discussion since we were ruling out polynomial query parallel algorithms.

The main observation is the strong property that until the algorithm queries a set $S$ with $S_A = R$, it obtains no information about the function $g$.
Therefore, if we can prove a lower bound $L(n,r)$ on the number of oracle queries any algorithm needs to find such a set, with $r$ being the size of $R$, then we can obtain an $\Omega(\frac{n}{r}\cdot L(n,r))$ lower bound  on the exact SFM query complexity. 

It is actually not too hard to prove $L(n,2) \geq \lfloor \log_2 n \rfloor - 1$ for any {\em deterministic} algorithm. Note that $R$ is a singleton element, and we overload notation and call that element $R$ as well.
First, note that for any query $S$, if $S_A \neq R$, then 
the value of $F(S)$ only reveals whether $S$ contains ``both'' the elements of $A$, ``none'' of the elements of $A$, or the ``other'' element in $A$ that is not $R$; in the first case, the $\phi$-function is negative, the second case it is positive and the last case it is $0$. The lower bound can now be proved using an {\em adversary} argument against the deterministic algorithm, by choosing the function so that the oracle never answers ``other.'' Since the algorithm is deterministic, the adversary can choose the set $A$ depending on the queries. The adversary maintains an ``active universe'' $U$ which initially contains all the elements.
If the first query $S$ contains $\leq |U|/2$ active elements, then the adversary puts both elements of $A$  in $V \setminus S$, answers ``none'', and removes $U\cap S$ from $U$; if $S$ contains $>|U|/2$ active elements then the adversary puts both elements in $S$, answers ``both'', and removes $U\setminus S$ from $U$. The algorithm can never reach the desired set until the number of active elements goes below $2$. Since the number of active elements can at best be halved each time, this proves a $\log_2 n - 1$ lower bound on the number of queries.
Together with our construction, we obtain an $\Omega(n\log n)$ lower bound on the query complexity of any deterministic SFM algorithm. This is the first super-linear lower bound for this question.

\medskip \noindent \textbf{Limitations and Open Questions.} We end this overview section by pointing out some limitations of our construction; we believe bypassing them would require new ideas.
The first issue is the {\em range} of our submodular functions. Our current way of constructing the submodularizer $\phi$ in \eqref{eq:def-phi} requires that the range of $\phi$ be distinctly smaller than the marginal increase in the $f_R$ function. This is noted by the parameter $\beta$ which is set to $\Theta(1/n)$. If there are $\ell = n/2r$ parts to the function, then due to the recursive nature of our construction, the smallest non-zero value our function takes is as small as $O(\frac{1}{n^\ell})$. 
When $\ell = \Theta(n/\log n)$, as is the case in our lower bound for parallel SFM, this is $2^{-\Theta(n)}$. Put differently, if we scale the function such that the range is integers, then our function's range takes exponentially large integer values. Therefore, our lower bounds are more properly interpreted in the {\em strongly polynomial} regime where the round/query-complexity needs to be independent of the 
range of the submodular function.
In contrast, the submodular functions constructed in~\cite{CCK21} which proves an $\Omgt(n^{1/3})$ lower bound on the number of rounds have range $\{-n, -n+1, \ldots, n-1, n\}$, 
and thus also constitute a lower bound in the {\em weakly polynomial} regime (its definition is deferred to \Cref{subsec:related-work}).
Interestingly, the lower bound construction in~\cite{BS20} also has a large range; it remains an interesting open problem to prove a nearly-linear lower bound on the number of rounds for query-efficient parallel SFM for integer-valued submodular functions with $\poly(n)$-bounded range. 

We prove an $\Omega(n\log n)$ lower bound for the query complexity of deterministic algorithms for SFM. Improving this to an $n^{1+c}$-lower bound for some constant $c > 0$ is an important open question.
The collection of functions we construct can be minimized in $\tilde{O}(n)$ queries, and so one may need new ideas to obtain a truly super-linear lower bound.
The main idea behind this algorithm is that in \eqref{eq:try2}, an element of $R$ can be recognized in $\polylog(n)$ queries using a binary-search style idea. 
Basically, given any set $S$ the function value $F(S)$ gives the information whether $S_A$ is a subset/superset of $R$ (in which case it also gives the size $|S_A|$), or it tells if $S_A$ is neither a subset or superset of $R$. With some work this leads to an $\tilde{O}(r)$ query algorithm to find $R$ (here $r$ is the size of $R$), and thus in $n/2r$ rounds with a total query complexity of $\tilde{O}(n)$ one minimizes $F$.

A final limitation is that we fall short of proving an $\Omega(n\log n)$ query lower bound for {\em randomized} SFM algorithms. Indeed, if one looks at the structure of our $\Omega(n\log n)$ proof, 
the ``$\log n$'' arises from $L(n,2)$ which is a lower bound on the number of queries a deterministic algorithm needs to make to find a set $S$ such that $S_A = R$.
With randomization, this problem is trivially solved in $O(1)$ queries; a random set that contains each element with probability $1/2$ would do. One may wonder if $r = |R|$ was increased, whether a super-linear in $r$ lower bound could be proved for $L(n,r)$. Unfortunately this is not possible; there is a randomized algorithm 
which finds a set $S$ with $S_A = R$ in expected $O(r)$ queries. 
We leave proving a super-linear lower bound on the query complexity of randomized algorithms for SFM as an open question. The family we construct is a potential candidate for the lower bound, just that a new technique would be needed to show this.

%%%%%%%%%%%%%%%%%%%%%%%%%%%%%%%%%
% Subsec: Further Related Work
%%%%%%%%%%%%%%%%%%%%%%%%%%%%%%%%%

\subsection{Further Related Work}\label{subsec:related-work}

\noindent \textbf{Other Regimes for SFM.}
Apart from the strongly-polynomial regime, 
there have also been multiple recent improvements to the complexity of SFM in other regimes that depend on $M$, the range of the function, i.e.\ $\max_{S \subseteq V} |f(S)|$ when $f$ is scaled to have an integer range. 
In particular, we refer to an algorithm as {\em weakly-polynomial} if the number of evaluation oracle queries it makes is polynomial in $n$ and $\log M$, and {\em pseudo-polynomial} if the number of queries is a polynomial in $n$ and $M$. 
State-of-the-art weakly-polynomial algorithms include $\tilde{O}(n^2\log M)$-query, $O(n^3 \cdot \poly(n, M))$-time algorithms \cite{LSW15,JLSW20}, and state-of-the-art pseudo-polynomial algorithms include $\tilde{O}(n \cdot \poly(M))$-query, $\tilde{O}(n \cdot \poly(M))$-time algorithms~\cite{CLSW17,ALS20}.

\medskip \noindent \textbf{Query Lower Bounds and Cuts.}
As far as the query complexity of SFM is concerned, lower bounds have been stagnating at $\Omega(n)$. The first known lower bound, of $n$ queries, is due to \cite{H08}. Motivated the problem of improving the lower bound, \cite{RSW17} considered graph cut functions, which is a subclass of submodular functions, and the problem of computing a global minimum cut in a graph using cut queries. However, they instead showed an upper bound of $\Tilde{O}(n)$ queries to find a (non-trivial) global minimum cut in an undirected, unweighted graph. 
\cite{GPRW19} improve the lower bound for SFM to $2n$ using an adversarial input technique, and also introduce a novel concept, called the graph cut dimension, for proving lower bounds for the min-cut settings. The main insight is that the cut dimension of a graph, defined as the dimension of the span of all vectors representing minimum cuts (binary vectors in $R^{E}$), is a lower bound on the number of cut queries needed. However, \cite{LLSZ20} has shown that the cut dimension of an unweighted graph is at most $2n - 3$, essentially eliminating the hope for a super-linear lower bound using this measure. Further, the recent work of \cite{AEG+22} provides a randomized algorithm that makes $O(n)$ queries and computes the global minimum cut in an undirected, unweighted graph with probability $2/3$.

\medskip \noindent \textbf{Parallel Convex Optimization.}
As far as parallel lower bounds are concerned, the general framework described in~\Cref{subsec:techniques} and employed in \cite{BS20, CCK21} is similar in spirit to 
the approach taken in~\cite{N94} to bound parallel non-smooth {\em convex} optimization. More precisely, \cite{N94} considers the problem of minimizing a non-smooth convex function $f$ (rescaled to be have range $[-1,+1]$) up to $\eps$-additive error in an $\ell_\infty$-ball, where one has access to first-order oracle and can make $\poly(n)$ queries to it in each round.
\cite{N94} 
shows that any query-efficient algorithm with parallel depth $\Ot(n^c \log(1/\eps))$ must have $c \geq 1/3$.

The proof relies on the idea of partitioning 
the universe $V$ into $r = \Omgt(n^{1/3} \log (1/\epsilon))$ parts, and considering functions $f$ that are the maximum of functions $f_i$ defined on these partitions.

\cite{BJL+19} uses a similar framework to 
show that any query-efficient algorithm achieving parallel depth $\Ot(n^c \log(1/\eps))$ must have $c \geq 1/2$.
\cite{N94} hypothesises that 
 such algorithms must have $c\geq 1$, but this is still open. 
The problem has also been studied~\cite{DBW12,BS18,DG19,BJL+19} when the dependence on $1/\eps$ is allowed to be a polynomial, and we refer the interested reader to these works for more details.

\medskip \noindent \textbf{Approximate SFM.}
Since the Lov\'asz extension of a submodular function is a non-smooth convex function, the discussion in the above paragraph is related to understanding the 
parallel complexity of $\eps$-{\em approximate} SFM. In this problem, we assume by scaling that the range of the function is in $[-1,+1]$ and the objective
is to obtain an additive $\eps$-approximation to the minimum value. 
The construction in~\cite{CCK21} 
shows that any query-efficient $\eps$-approximate SFM algorithm 
with  depth $\Ot(n^c\log(1/\eps))$ must have $c \geq 1/3$.
Note the similarity with the lower bound in~\cite{N94} mentioned in the previous paragraph; this is not an accident since the bottlenecks due to standard deviation considerations are similar in both approaches.
A reader may wonder if the constructions in our paper 
also prove that any  query-efficient $\eps$-approximate SFM algorithm with  depth $\Ot(n^c\log(1/\eps))$ must have $c \geq 1$. 
This is not the case; the functions we consider can be $\eps$-approximated in $O(\log(1/\eps))$-rounds. This stems from the limitation in our construction that the ``scale'' of the functions we consider across the layers decay geometrically, and thus one can get $\eps$-close in $O(\log(1/\eps))$-rounds.

The $\eps$-approximate SFM question is also interesting when the dependence of the depth on $1/\eps$ is allowed to be a polynomial. 
In this setting, one can leverage the parallel convex optimization works mentioned in the previous paragraph to obtain query-efficient $\eps$-approximate SFM algorithms with depth being
truly sub-linear in $n$. 
For instance, the algorithm in \cite{BJL+19} implies a query-efficient $\eps$-approximate SFM algorithm running in $\Ot(n^{2/3}\eps^{-2/3})$-rounds. 
On the other hand, the construction in~\cite{CCK21} shows that any query-efficient $\eps$-approximate SFM algorithm with depth 
$(1/\eps)^c$ must have $c\geq 1$.
Understanding the correct answer for query-efficient $\eps$-approximate SFM, both when the dependence on $\eps$ is $\poly(1/\eps)$ and when it is $\log(1/\eps)$, is an interesting open question.

\section{Preliminaries}

Throughout, $\log$ denotes logarithm with base $2$. 
For any two sets $X$ and $Y$, we use $X \subseteq Y$ to denote that $X$ is a subset of $Y$ with possibly $X = Y$; we use $X \subsetneq Y$ to denote that $X$ is a strict subset of $Y$, i.e. $X \subseteq Y$ and there exists at least one element $e \in Y$ such that $e \notin X$. Further, supersets, $\supseteq$, and strict supersets, $\supsetneq$, are defined analogously. 

For any set $X$ and element $e \notin X$, we let $X + e$ denote the set obtained by including $e$ into $X$, i.e. $X \cup \{e\}$.
Given two sets $X$ and $Y$, we define $Y \setminus X = \{e \in Y: e \notin X\}$ to denote the set of elements in $Y$ but not in $X$.

\begin{definition}[Marginals] \label{defn:marginals}
Let $f: 2^V \rightarrow \mathbb{R}$ for finite set $V$. For any $X \subsetneq V$ and $e \in V \setminus X$, we define $\partial_e f(X) := f(X + e) - f(X)$, the marginal of $f$ at $X$ when adding element $e$. 
\end{definition}

\begin{definition}[Submodular functions] \label{defn:submodular_functions} 
A set function $f: 2^V \rightarrow \mathbb{R}$ for finite set $V$ is \emph{submodular} if $\partial_e f(Y) \leq \partial_e f(X)$, for any subsets $X \subseteq Y \subsetneq V$ and $e \in [n] \setminus Y$. An alternative definition is that for any two subsets $X, Y \subseteq V$, the following inequality holds
\begin{equation}\label{eq:defsbm}
    f(X) + f(Y) \geq f(X \cup Y) + f(X \cap Y) .
  \end{equation}
\end{definition}

%%%%%%%%%%%%%%%%%%%%%%%%%%%%%%%%%
% Sec: Our Construction
%%%%%%%%%%%%%%%%%%%%%%%%%%%%%%%%%

\section{Our Construction}
\label{sec:construction}

In this section, we describe our recursive construction of the family of non-negative functions $\mathcal{F}_r(V)$ on subsets of a given set of elements $V$, where $r \in \mathbb{Z}_+$ is an integer such that $2r$ divides $|V|$.  
We prove that any function $F \in \mathcal{F}_r(V)$ is submodular and its unique minimizer takes a special partition structure which is crucial to our proofs of lower bounds in \Cref{sec:lower_bounds}. 

We define the main building block behind our construction in \Cref{subsec:building_block}, and use it to recursively construct the function family $\mathcal{F}_r(V)$ in \Cref{subsec:function_family}.

%%%%%%%%%%%%%%%%%%%%%%%%%%%%%%%%%
% Subsec: Building Block
%%%%%%%%%%%%%%%%%%%%%%%%%%%%%%%%%

\subsection{Main Building Block}
\label{subsec:building_block}

We start by describing the main building block for our construction, which relies on two components. 
The first component is a standard submodular function corresponding to the sum of the rank functions of two rank-$1$ matroids \cite{H08,CLSW17}. 
The second component is a ``submodularizer'' function $\phi$. Despite not being submodular itself, this submodularizer function guarantees the submodularity of our main building block function.

%%%%%%%%%%%%%%%%%%%%%%%%%%%%%%%%%
% Paragraph: Component I: Intersection of Rank-1 Matroids
%%%%%%%%%%%%%%%%%%%%%%%%%%%%%%%%%
\medskip
\noindent \textbf{Component I: Sum of Two Rank-1 Matroids.} 
For any sets $R \subseteq A$,
we define the function $f_{A,R}: 2^{A} \rightarrow \mathbb{R}$ as
\begin{align}
f_{A, R}(S) := \begin{cases}
0 \qquad &\text{if $S = R$,} \\
1 \qquad &\text{if $S \subsetneq R$ or  $S \supsetneq R$,} \\
2 \qquad &\text{otherwise}.
\end{cases} \label{eq:compI} 
\end{align}

As noted in \cite{H08}, the function $f_{A,R}$ above corresponds to the matroid intersection of two rank-$1$ matroids, and is therefore submodular.

\begin{lemma}[\cite{H08}] \label{lem:submodularity_f}
For any $R \subseteq A$, the function $f_{A,R}: 2^A \rightarrow \mathbb{R}$ defined above is submodular.
\end{lemma}

In fact, the submodular function $f_{A,R}$ (appropriately scaled) has previously been used in \cite{H08} to prove an $n$ lower bound on the number of evaluation oracle calls, and in \cite{CLSW17} to show an $n/4$ lower bound on the number of sub-gradients of the Lov\'asz extension for SFM.

%%%%%%%%%%%%%%%%%%%%%%%%%%%%%%%%%
% Paragraph: Component II: The Submodularizer
%%%%%%%%%%%%%%%%%%%%%%%%%%%%%%%%%

\medskip
\noindent \textbf{Component II: The Submodularizer.}
Let $R \subseteq A \subseteq V$ be subsets of the ground set $V$, and denote $B := V \setminus A$. For any subset $S \subseteq V$, we denote $S_A := S \cap A$ and $S_B := S \cap B$. 

Ideally, we would like to recursively define a function on $V$ to be of the form $f_{A,R}(S_A) + \bone(S_A = R) \cdot g(S_B)$, where $g: 2^B \rightarrow \R$ is a submodular function on $B$. However, as mentioned in \Cref{subsec:techniques}, such a function may not be submodular even when both $f_{R,A}$ and $g$ are submodular.
For our recursive construction to go through,  we define the following submodularizer function:
$\phi_{V,A,R}: 2^V \rightarrow \mathbb{R}$ as
\begin{align}
\phi_{V,A,R}(S) := 
\begin{cases}
|S_B| \qquad &\text{if $S_A \subsetneq R$,} \\
-|S_B| \qquad &\text{if $S_A \supsetneq R$,} \\
0 \qquad &\text{otherwise, and in particular when $S_A = R$}.
\end{cases} \label{eq:compII}
\end{align}

Note that the function $\phi_{V,A,R}$ defined above is not submodular, as witnessed by the following violation of the marginal property in \Cref{defn:submodular_functions}. 
To see this, let $X \subseteq Y \subseteq V$ be any two subsets such that 
$X_A = R$, $A \neq Y_A \supsetneq X_A$, and $X_B \neq \emptyset$. Note that $Y_A$ is a strict superset of $X_A$.
Pick an element $e \in A \setminus Y_A$. 
Then observe that $\partial_e \phi_{V,A,R}(X) = - |X_B| < 0$ since $\phi_{V,A,R}(X\cup e) = -|X_B|$ and $\phi_{V,A,R}(X) = 0$.
On the other hand, both $\phi_{V,A,R}(Y\cup e) = \phi_{V,A,R}(Y) = -|Y_B|$ implying
$\partial_e \phi_{V,A,R}(Y) = 0 > \partial_e \phi_{V,A,R}(X)$. This is a violation of submodularity.
However, these are the only cases where submodularity is violated, and  it turns out that this ``almost submodularity'' property helps to guarantee the submodularity of our main building block which we define next. 

\medskip
\noindent \textbf{The main building block.} 
Let $R \subseteq A \subseteq V$ be non-empty subsets of a finite set $V$ and denote $B := V \setminus A$. 
Let $g: 2^B \rightarrow \mathbb{R}$ be a set function on $B$ and $M \geq 0$ be a parameter such that $\max_{S \subseteq B} |g(S)| \leq M$. 
Our main building block is the function $F^{M,g}_{V,A,R}: 2^V \rightarrow \mathbb{R}$ defined as 
\begin{align} \label{eq:building_block}
F^{M,g}_{V,A,R}(S) := f_{A,R}(S \cap A) + \frac{1}{2|V|} \cdot \phi_{V,A,R}(S) + \frac{1}{4M|V|} \cdot \ind(S_A = R) \cdot g(S \cap B) .
\end{align}
The function $F^{M,g}_{V,A,R}$ will be used in \Cref{subsec:function_family} to construct a function family on $V$ by choosing $g$ from the function family recursively defined on $B$. 
To show the submodularity and structural properties of minimizers of this recursive constructed function family, we first prove the following properties of the function $F^{M,g}_{V,A,R}$.

\begin{restatable}[Properties of main building block]{lemma}{PropertyBuildingBlock} \label{lem:properties_building_block}
Let $V$ be a finite set of elements, $R \subseteq A \subseteq V$ be non-empty subsets of $V$, and denote $B := V \setminus A$. Let $g: 2^B \rightarrow \mathbb{R}$ be a submodular function taking values in $[0,M]$ that has a unique minimizer $S_g^* \subseteq B$. Then the function $F := F^{M,g}_{V,A,R}$ defined in \eqref{eq:building_block} satisfies the following properties:
\begin{enumerate}
    \item (Non-negativity and boundedness) For any subset $S \subseteq V$, we have $F(S) \in [0,2]$, 
    \item (Unique Minimizer) $F$ has a unique minimizer $R \cup S_g^*$, 
    \item (Submodularity) $F$ is submodular. 
\end{enumerate}
\end{restatable}

As mentioned in \Cref{subsec:techniques}, the main insight behind the proof of  \Cref{lem:properties_building_block} is that the scale of the function $\frac{1}{4M |V|} \cdot \ind(S_A = R) \cdot g(S_B)$ is smaller than that of $\frac{1}{2|V|} \cdot \phi_{V,A,R}(S)$, and both are much smaller than that of $f_{A,R}$. As such, the minimizer $S^*$ and the range of $F^{M,g}_{V,A,R}$ are dominantly determined by the function $f_{A,R}$, enforcing $S^*_A = R$ and thus $f_{A,R}(S^*_A) = \phi_{V,A,R}(S^*) = 0$. 
Moreover, most cases where submodularity fails to hold for the function $\frac{1}{4M |V|} \cdot \ind(S_A = R) \cdot g(S_B)$ can be corrected by the submodularizer $\frac{1}{2|V|} \cdot \phi_{V,A,R}(S)$, and the very few cases where submodularity fails to hold for $\frac{1}{2|V|} \cdot \phi_{V,A,R}(S)$ can be fixed by the dominant submodular function $f_{A,R}$. 
We postpone a formal proof of \Cref{lem:properties_building_block} to \Cref{sec:missing_proofs}.

%%%%%%%%%%%%%%%%%%%%%%%%%%%%%%%%%
% Subsec: The Function Family
%%%%%%%%%%%%%%%%%%%%%%%%%%%%%%%%%

\subsection{The Function Family}
\label{subsec:function_family}

Using our main building block described in \Cref{subsec:building_block}, we now define the function family $\mathcal{F}_r(V)$ recursively for all finite sets $V$ with $|V|$ divisible by $2r$. 

\smallskip
\noindent \textbf{The base case: when $|V| = 2r$.} In this case, we let $\mathcal{F}_r(V) := \{f_{V,R}: R \subseteq V, |R| = r\}$.  

\smallskip
\noindent \textbf{Recursive definition.}  
Suppose the function family $\mathcal{F}_r(V)$ has been defined for all $|V| = 2r(k-1)$ for integer $k \geq 2$, we now define the family $\mathcal{F}_r(V)$ for $|V| = 2rk$ as follows: 
\begin{align*}
\mathcal{F}_r(V) := \{F^{2,g}_{V,A,R}: R \subseteq A \subseteq V, |R| = |A|/2 = r, g \in \mathcal{F}_r(V\setminus A)\} ,
\end{align*}
where we recall from \eqref{eq:building_block} that \begin{align} \label{eq:recursive_defn}
    F^{2,g}_{V,A,R} = f_{A,R}(S_A) + \frac{1}{2|V|} \cdot \phi_{V,A,R}(S) + \frac{1}{8|V|} \cdot \ind(S_A = R) \cdot g(S_B) .
\end{align}
This completes the recursive definition of the family of functions $\mathcal{F}_r(V)$, where $|V|$ is divisible by $2r$. 
When $|V|$ is not a multiple of $2r$, we may also naturally extend the definition above by making $|V| - 2r \cdot \lfloor \frac{|V|}{2r} \rfloor$ elements ``dummy'' in $V$. More precisely, we let $V' \subseteq V$ be an arbitrary subset with size $|V'| = 2r \cdot \lfloor \frac{|V|}{2r} \rfloor$, and define the function family to only depend on elements in $V'$.

%%%%%%%%%%%%%%%%%%%%%%%%%%%%%%%%%
% Subsubsec: Explicit Formula
%%%%%%%%%%%%%%%%%%%%%%%%%%%%%%%%%

\def\cA{\mathcal{A}}
\def\cR{\mathcal{R}}

\medskip
\noindent \textbf{Explicit Formula for Our Construction.}
We give more explicit expressions for functions in $\mathcal{F}_r(V)$ recursively defined above, assuming $|V|$ is divisible by $2r$. 
Let $\ell := |V| / 2r $, and consider any partition $\mathcal{A}$ of the universe $V = A_1 \cup A_2 \cup \cdots \cup A_\ell$, where $|A_i| = 2r$ for all $i \in [\ell]$.
Furthermore, we select subsets
$R_i \subseteq A_i$ for each $i \in [\ell]$ with size $|R_i| = r$. Let $\mathcal{R}$ denote the collection of these $R_i$'s. 
We denote $B_i := \cup_{j=i}^\ell A_j = V \setminus (\cup_{j=1}^{i-1} A_j)$ the remaining set of elements when $A_1, \cdots, A_{i-1}$ are removed from $V$.
Given the partition $\cA$ and the family of subsets $\cR$, we
define a function $F_{\cA,\cR}: 2^V \rightarrow \mathbb{R}$ as follows. 
For any $S \subseteq V$, let $k_S$ be the smallest index $k \in [\ell]$ such that $S_{A_k} := S \cap A_k \neq R_k$.
If such an index $k_S$ does not exist, that is $S\cap A_k = R_k$ for all $k\in [\ell]$,  then we set $F_{\cA,\cR}(S) := 0$. Otherwise, we define its value
\begin{align} \label{eq:explicit_defn}
F_{\cA,\cR}(S) := \left(\prod_{j = 0}^{k_S-2} \frac{1}{8(|V| - 2 j r)} \right) \cdot \left(  f_{A_{k_S}, R_{k_S}}(S_{A_{k_S}}) + \frac{1}{2|B_{k_S}|} \cdot \phi_{B_{k_S}, A_{k_S}, R_{k_S}}(S_{B_{k_S}}) \right) 
\end{align}
where $f_{A_{k_S}, R_{k_S}}$ and $\phi_{B_{k_S}, A_{k_S}, R_{k_S}}$ as defined in \eqref{eq:compI} and \eqref{eq:compII}.

We now claim that the function family $\mathcal{F}_r(V)$ defined above coincides with the collection of all functions $F_{\cA,\cR}$, for all partitions $V = A_1 \cup A_2 \cup \cdots \cup A_\ell$ with $|A_i| = 2r, \forall i \in [\ell]$ and subsets $R_i \subseteq A_i$ with $|R_i| = r, \forall i \in [\ell]$.
To see why this is the case, note that in \eqref{eq:recursive_defn}, the functions $f_{A_j, R_j}(S_{A_j}) = \phi_{B_j, A_j, R_j}(S_{B_j}) = 0$ for all $j \leq k_S - 1$, and the indicator $\ind(S_{A_{k_S}} = R_{k_S}) = 0$. It follows that the functions $f_{A_{k_S}, k_S}$ and $\phi_{B_{k_S}, A_{k_S}, R_{k_S}}$ are the only non-zero components when we expand out the recursive part $g$ in \eqref{eq:recursive_defn}.

The explicit expression \eqref{eq:explicit_defn} reveals important insights into why functions in $\mathcal{F}_r(V)$ take a large number of rounds to minimize. 
Roughly speaking, any query $S$ would only reveal information about the subsets $R_j \subseteq A_j$ for $j \leq k_S$, but nothing about subsets $R_j \subseteq A_j$ for any $j \geq k_S + 1$. 
If in each round of queries, an algorithm advances $k_S$ by at most $1$, then obtaining full information about the function $F_{\{A_i\},\{R_i\}}$ requires at least $n/2r$ rounds of queries.

%%%%%%%%%%%%%%%%%%%%%%%%%%%%%%%%%
% Subsubsec: Properties of Our Construction
%%%%%%%%%%%%%%%%%%%%%%%%%%%%%%%%%

\subsubsection{Properties of Our Construction}

The following lemma collects properties of the function family $\mathcal{F}_r(V)$. In particular, any function $F \in \mathcal{F}_r(V)$ is submodular, and its unique minimizer admits a partition structure. 
These properties follow from the corresponding properties of our main building block proved in \Cref{lem:properties_building_block}

\begin{lemma}[Properties of our construction] \label{lem:properties_function_family}
Let $V$ be a finite set of elements and $r \in \mathbb{Z}_+$ satisfies $2r$ divides $|V|$. 
Then any function $F \in \mathcal{F}_r(V)$ satisfies the following properties:
\begin{enumerate}
    \item (Non-negativity and boundedness) For any subset $S \subseteq V$, we have $F(S) \in [0,2]$, 
    \item (Unique Minimizer) $F$ has a unique minimizer of the form $S^* = \cup_{i=1}^\ell R_i$, where $V = A_1 \cup \cdots \cup A_\ell$ forms a partition with $\ell = |V|/2r$ and $|A_i| = 2r, \forall i \in [\ell]$, and subsets $R_i \subseteq A_i$ have size $|R_i| = r, \forall i \in [\ell]$, 
    \item (Submodularity) $F$ is submodular. 
\end{enumerate}
\end{lemma}

\begin{proof}
We prove the lemma by induction based on the size of the ground set $V$.

\smallskip
\noindent \textbf{The base case.} The base case is when $|V| = 2r$ and the statement in this case follows because the function $f_{V,R}$ has range $\{0,1,2\}$, unique minimizer $R$ and is submodular by \Cref{lem:submodularity_f}. 

\noindent \textbf{The induction step.}
Suppose we have proven the three properties of the lemma when the size of the ground set is $2r(k-1)$ for some $k \geq 2$, we now prove the three properties for $|V| = 2rk$.  

Note that any function $F \in \mathcal{F}_r(V)$ takes the form 
\begin{align*}
F(S) = F^{2,g}_{V,A,R}(S) = f_{A,R}(S_A) + \frac{1}{2|V|} \cdot \phi_{V,A,R}(S) + \frac{1}{8|V|} \cdot \ind(S_A = R) \cdot g(S_B) . 
\end{align*} 
for some subsets $R \subseteq A \subseteq V$ such that $|R| = |A|/2 = r$, and function $g \in \mathcal{F}_r(B)$ with $B = V \setminus A$.
By induction hypothesis, $g$ satisfies the three properties in the lemma. 
The three properties for function $F$ then follows immediately from applying  \Cref{lem:properties_building_block} with $M=2$. 
\end{proof}

%%%%%%%%%%%%%%%%%%%%%%%%%%%%%%%%%
% Section: Lower Bounds
%%%%%%%%%%%%%%%%%%%%%%%%%%%%%%%%%

\section{Lower Bounds}
\label{sec:lower_bounds}

In this section, we leverage our construction of the function family $\mathcal{F}_r(V)$ from \Cref{sec:construction} to prove lower bounds for SFM. In \Cref{sec:lower-det}, we prove an $\Omega(n \log n)$ evaluation query complexity lower bound for any deterministic algorithm that minimizes functions in $\mathcal{F}_r(V)$, even when $r = 1$.
Then, in \Cref{sec:lower-parallel}, 
we show that any randomized parallel SFM algorithm that makes at most $Q = \poly(n)$ evaluation oracle queries per round, with high probability, takes at least $\Omega(n / \log n)$ rounds to minimize a uniformly random function $F \in \mathcal{F}_r(V)$ for $r = \Theta(\log n)$.

%%%%%%%%%%%%%%%%%%%%%%%%%%%%%%%%%
% Subsec: Oracle Lower Bound
%%%%%%%%%%%%%%%%%%%%%%%%%%%%%%%%%

\subsection{Query Complexity Lower Bound for Deterministic Algorithms}
\label{sec:lower-det}

In this subsection, we prove the query complexity lower bound for deterministic SFM algorithms in \Cref{thm:oracle_complexity_lb_det}, with the function $F$ chosen adversarially from the function family $\mathcal{F}_1(V)$. 
More specifically, we prove the following theorem which immediately implies \Cref{thm:oracle_complexity_lb_det}. 

\begin{theorem}[Query complexity lower bound for deterministic algorithms] \label{thm:oracle_complexity_lb_det_body}
Let $V$ be a finite set with $n$ elements. 
For any deterministic SFM algorithm $\ALG$, there exists a submodular function $F \in \mathcal{F}_1(V)$ such that $\ALG$ makes  at least $\frac{n}{2} \log_2 (\frac{n}{4})$ evaluation oracle queries to minimize $F$. 
\end{theorem}

Let us fix a deterministic algorithm $\ALG$. We prove that there exists a function $F\in \mathcal{F}_1(V)$ on which $\ALG$ must make at least $\frac{n}{2}\log \left(\frac{n}{4}\right)$ 
evaluation oracle queries. From~\eqref{eq:recursive_defn}, recall that any function $F \in \mathcal{F}_1(V)$ is specified by subsets $R \subseteq A \subseteq V$ where $|A| = 2$ and $|R| = 1$, and a function $g \in \mathcal{F}_1(B)$, where $B := V \setminus A$. As $R$ contains only a single element and we abuse notation and call that element $R$ as well.
The function $F$ is then given by $F(S) := f_{A,R}(S_A) + \frac{1}{2|V|} \cdot \phi_{V,A,R}(S) + \frac{1}{8|V|} \cdot \ind(S_A = R) \cdot g(S_B)$. Recall $S_A$ is the shorthand for $S\cap A$
and $S_B$ is the shorthand for $S\cap B$.
By \Cref{lem:properties_function_family}, $F(S)$ has a unique minimizer $S^*$ with $S^*_A = R$ and $S^*_B$ is the unique minimizer of $g(S_B)$. 

By construction, until $\ALG$ queries a set $S$ with $S_A = R$, that is, $S\cap A$ is precisely the singleton $R$, it obtains no information about $g$.
More precisely, the answers given to $\ALG$ are the same no matter which $g\in \mathcal{F}_1(B)$ is picked. The heart of the lower bound is the following lemma
which asserts that an adversary can always choose an $(A,R)$ pair such that the first $O(\log n)$-queries of $\ALG$ ``miss $R$'', that is, $S_i\cap A \neq R$.

\begin{lemma} \label{claim:oracle_lb}
	Fix a deterministic algorithm $\ALG$ and let $T := \lfloor \log n \rfloor - 1$. 
	There exist $R \subseteq A \subseteq V$ with $|R| = 1$ and $|A| = 2$ such that the first $T$ (possibly adaptive) queries $S^1, \cdots, S^T$ made by $\ALG$  to the evaluation oracle $\EO$ satisfy $S^i_A \neq R$ for all $i \in [T]$. 
\end{lemma}

Before we prove the above lemma, let us first use it to prove~\Cref{thm:oracle_complexity_lb_det}.

\begin{proof}[{\bf Proof of~\Cref{thm:oracle_complexity_lb_det}}]
Fix a deterministic algorithm $\ALG$. 
For any even integer $n \geq 2$, let $h(n)$ denote the smallest integer such that $\ALG$ makes at most $h(n)$ oracle calls to minimize any submodular function $F \in \mathcal{F}_1(V)$ with $|V| = n$, 
even when $\ALG$ is given the information that the submodular function is picked from this family.
We claim that $h(n) \geq \frac{n}{2} \log (\frac{n}{4})$. Since by~\Cref{lem:properties_function_family}, any function $F \in \mathcal{F}_1(V)$ is submodular, this would imply~\Cref{thm:oracle_complexity_lb_det}.
We prove the claim by induction; the base case of $n=2$ holds vacuously.

Let $T = \lfloor \log n \rfloor - 1$. 
By \Cref{claim:oracle_lb}, we can choose subsets $R \subseteq A \subseteq V$ such that $|R| = 1$, $|A| = 2$, and for the first $T$ (possibly adpative) queries $S^1, \ldots, S^T$ of $\ALG$,
we have $S^i_A \neq R$ hold for all $i \in [T]$. Now consider the function $F\in \mathcal{F}_1(V)$ defined as  
\[F(S) := f_{A,R}(S_A) + \frac{1}{2|V|} \cdot \phi_{V,A,R}(S) + \frac{1}{8|V|} \cdot \ind(S_A = R) \cdot g(S_B),\] where $(A,R)$ are these subsets, $B = V\setminus A$, 
and $g$, by induction, is the function in $\mathcal{F}_1(B)$ on which $\ALG$ takes $h(n-2)$ queries (since $|B| = |V| - 2$) to find the unique minimizer.
By the choice of $(A,R)$, since $S^i_A \neq R$, the evaluations of $F(S^i)$ are the same {\em for all} $g\in \mathcal{F}_1(B)$. In other words,
in its first $T = \lfloor \log n \rfloor - 1$ queries, $\ALG$ does not obtain any information about the function $g$. 

After $T$ queries, suppose we  provide $\ALG$ with $(A,R)$. By~\Cref{lem:properties_function_family}, $\ALG$ now needs to minimize $g$.
Since the answers received by $\ALG$ are consistent with any $g\in \mathcal{F}_1(B)$, by induction, $\ALG$ takes at least $h(n-2)$ queries to minimize $g$.
Therefore, we get the recursive inequality $h(n) \geq h(n-2) + \lfloor \log n \rfloor - 1$. This implies $h(n) \geq \frac{n}{2} \log (\frac{n}{4})$. 
proving the theorem statement. \end{proof}

Now we are left to prove \Cref{claim:oracle_lb}. 

\begin{proofof}{\Cref{claim:oracle_lb}}
The proof is via an adversary argument where the $\EO$ is an adversary trying to foil the deterministic algorithm $\ALG$. 
In particular, $\EO$ can choose to not commit to the sets $(A,R)$ in the definition of the function $F\in \mathcal{F}_1$ at the beginning. Instead, at every query $S^i$,
the adversary oracle $\EO$ gives an answer consistent with a function $F(S) = f_{A,R}(S_A) + \frac{1}{2|V|} \cdot \phi_{V,A,R}(S) + \ind(S_A = R)g(S_B)$ for some $(A,R)$
such that $S^i_A \neq R$ and such that all previous query answers are also consistent with $S$. We now show that this is possible for the first $T$ queries.

It is in fact convenient to consider the following modified evaluation oracle $\EO'$. When queried with a set $S \subseteq V$, $\EO'$ returns the following information: (1) whether $S_A = R$, or $S_A \subsetneq R$, or $R \subsetneq S_A$, or if $S_A$ is neither a subset nor a superset of $R$, and (2) the size of $|S_A|$. 
Note that unless $S_A = R$, the information returned by $\EO'$ is enough for the algorithm to compute $F(S)$. 
Indeed, when $S_A \neq R$, the function $F(S) = f_{A,R}(S_A) + \frac{1}{2|V|} \cdot \phi_{V,A,R}(S)$ so the information in (1) and (2), together with $|S|$ determine the value of $F(S)$. 
In short, we can use $\EO'$ to simulate $\EO$ till a query $S$ with $S_A = R$ is made. We now show how to construct the adversary $\EO'$ such that in the first $T$ queries, 
it can give answers such that $S^i_A \neq R$ for all $i\in [T]$ and there exists an $R\subseteq A\subseteq V$ consistent with all answers given so far.

The adversary $\EO'$ maintains an active set $U^1$ of elements which is initialized to $V$. Consider the first query $S^1$ made by $\ALG$.
If $|U^1 \cap S^1| \geq |U^1|/2$, then $\EO'$ does the following: (a) it sets $U^2 \leftarrow U^1\cap S^1$, and (b) answers $S^1_A = A$, that is, $R\subsetneq S^1_A$ and $|S^1_A| = 2$.
If $|U^1 \cap S^1| < |U^1|/2$, then $\EO'$ does the following: (a) it sets $U^2 \leftarrow U^1\setminus S^1$, and (b) answers $S^1_A = \emptyset$, that is, $R\supsetneq S^1_A$ and $|S^1_A| = 0$.
In short, the adversary $\EO'$ commits that $A\subseteq U^2$, and for {\em any} such $A$ and any $R\subseteq A$, the answer given above would be consistent.

More generally, at the beginning of round $i$, the adversary $\EO'$ has an active set $U^i$ with $\geq 4$ elements. Upon query $S^i$, 
if $|U^i\cap S^i| \geq |U^i|/2$, then $\EO'$ answers $R\subsetneq S^i_A$ and $|S^i_A| = 2$, and modifies $U^{i+1} \leftarrow U^i \cap S^i$, 
otherwise, $\EO'$ answers $R\supsetneq S^i_A$ and $|S^i_A| = 0$, and modifies $U^{i+1} \leftarrow U^i \setminus S^i$.

Since the size of $U^i$ can at most halve, at the end of $T = \lfloor \log_2(n) \rfloor - 1$ rounds, the adversary $\EO'$ ends up with a set
$U^{T+1}$ with $\geq 2$ elements. At this point, $\EO'$ can choose any subset $R\subseteq A \subseteq U^{T+1}$ with $|A| = 2$ and $|R|=1$, 
and (a) all answers given above are consistent, and (b) $S^i_A \neq R$ for all $i\in [T]$. This completes the proof of the lemma.
\end{proofof}

\begin{remark}
	We note that~\Cref{claim:oracle_lb} is false if $\ALG$ is allowed to be randomized. Indeed, if $|A| = 2$ and $R\subseteq A$ has $|R| = 1$, then 
	any query $S$ which picks every element with probability $1/2$ will satisfy $S_A = R$ with probability $1/4$. Therefore, the proof idea breaks down for 
	randomized algorithms. On the other hand, we do not know of a randomized algorithm for minimize functions in $\mathcal{F}_1(V)$ that makes $O(n)$ queries
	and succeeds with constant probability.
\end{remark}

%%%%%%%%%%%%%%%%%%%%%%%%%%%%%%%%%
% Subsec: Parallel Lower Bound
%%%%%%%%%%%%%%%%%%%%%%%%%%%%%%%%%

\subsection{Parallel Lower Bound for Randomized Algorithms}
\label{sec:lower-parallel}

In this subsection, we prove the $\Omega(n/C \log n)$-lower bound on the number of rounds for (possibly randomized) parallel SFM algorithms in \Cref{thm:adaptive_lb}.  
By Yao's minimax principle, \Cref{thm:adaptive_lb} is implied by the following theorem where the function $F$ is chosen uniformly at random from the family $\mathcal{F}_r(V)$ with $r = C \log n$.

\begin{theorem}[Parallel lower bound for randomized algorithms] \label{thm:adaptive_lb_body}
Let $C \geq 2$ be any constant. Let $V$ be a finite set with $n$ elements, and $r \geq C \log n $ be an integer such that $2r$ divides $n$. Then 
any  parallel algorithm that makes at most $Q := n^C$ queries per round, and runs for $< (n/2r)$ rounds, fails to minimize a uniformly random submodular function $F \in \mathcal{F}_r(V)$, with high probability.
\end{theorem}

\begin{proof}
By the recursive construction of the function family $\mathcal{F}_r(V)$ in \Cref{subsec:function_family}, we may view a random submodular function $F$ drawn from the uniform distribution over $\mathcal{F}_r(V)$ being obtained as follows. 
We first select a uniformly random subset $A_1 \subseteq V$ of size $|A_1| = 2r$ and a uniformly random subset $R_1 \subseteq A_1$ with size $|R_1| = r$. 
Denoting $B := V \setminus A_1$, we then draw a uniformly random function $g \in \mathcal{F}_r(B)$, and let $F(S) := f_{A_1,R_1}(S_{A_1}) + \frac{1}{2|V|} \cdot \phi_{V,A_1,R_1}(S) + \frac{1}{8|V|} \cdot \ind(S_{A_1} = R_1) \cdot g(S_B)$. 
Coupled with $F(S)$ in terms of the randomness of the subsets $A_1$ and $R_1$, we also let $F'(S) := f_{A_1,R_1}(S_{A_1}) + \frac{1}{2|V|} \cdot \phi_{V,A_1,R_1}(S)$.

Since we have specified a distribution over submodular functions, it suffices to prove that any deterministic algorithm which runs in $< \frac{n}{2r}$ rounds and makes $\leq n^C$ queries per round, fails to find the minimizer 
of $F$ with high probability. 
In the remainder we prove this statement.

Consider the set of queries $S^1_1, \cdots, S^Q_1$ made by a deterministic algorithm $\ALG$ in the first round. We start by showing that with high probability, $S^i_1\cap A_1 \neq R_1$ for all $i \in [Q]$. 
This is because for any $S^i_1$ and any fixed outcome of $A_1$, since $R_1$ is a uniformly random subset of $A_1$ with size $r$, there are $\binom{2r}{r} \geq \frac{2^{2r}}{2r+1} \geq \frac{n^{2C}}{2 C \log n + 1}$ possible choices of $R$. 
Therefore, for any query $S^i_1$ and any fixed outcome of $A_1$, the probability that $S^i_1 \cap A_1 = R_1$ is at most $\frac{2 C \log n + 1}{n^{2C}}$.
It then follows by a union bound over all $S^i_1$ that with probability at least $1 - \frac{2 C \log n + 1}{n^C}$, the event $\mathcal{E}_1 := \{S^i_1 \cap A_1 \neq R_1, \forall i \in [Q]\}$ holds.

Now conditioning on the event $\mathcal{E}_1$, the output of the evaluation oracle when queried with $S^i_1$ would be $F(S^i_1) = F'(S^i_1)$, for all $i \in [Q]$. Note, however, that the function $F'$ does not depend on the randomness of $g \in \mathcal{F}_r(B)$. Thus, even when given the information of $R$ and $A$ after the first round of queries, $\ALG$ does not obtain any information about the uniformly random function $g \in \mathcal{F}_r(B)$. Therefore, we can apply the argument in the previous paragraph to the set of queries $S^1_2, \ldots, S^Q_2$ in the second round of the algorithm.
In particular, with probability at least $1 - 1/n^C$, the event $\mathcal{E}_2 := \{S^i_2 \cap A_2 \neq R_2, \forall i \in [Q]\}$ holds. 

More generally, if the algorithm makes $k < n/2r$ rounds of queries, then with probability $\geq 1 - \frac{k(2 C \log n + 1)}{n^C} > 1 - \frac{1}{n^{C-1}}$ all the events $\mathcal{E}_i$ occur.
This implies that the answers obtained by the algorithm are consistent with any function in $\mathcal{F}_r(V)$ where
the sets $A_1, \ldots, A_k$ and $R_1, \ldots, R_k$ are fixed, but the sets $A_{k+1}, \ldots, A_{n/2r}$ and $R_{k+1}, \ldots, R_{n/2r}$ are completely random. Since the unique minimizer of $F$ is the set $(R_1 \cup R_2 \cup \cdots \cup R_{n/2r})$, no matter which set the deterministic algorithm returns, it will err with probability at least $1 - \frac{1}{n^{C-1}}$. This completes the proof of the theorem.
\end{proof}

\section{Acknowledgements}
We thank the anonymous reviewers of FOCS 2022 for helpful comments. Part of this work was done while Deeparnab Chakrabarty, Haotian Jiang and Aaron Sidford were attending the Discrete Optimization trimester program at the Hausdorff Research Institute for Mathematics. 

Deeparnab Chakrabarty is supported by NSF grant 2041920. 
Andrei Graur is supported by NSF CAREER Award CCF-1844855, NSF Grant CCF-1955039, and Stanford's Nakagawa Fellowship. 
Haotian Jiang is supported by NSF awards CCF-1749609, DMS-1839116, DMS-2023166, CCF-2105772, and a Packard Fellowship. 
Aaron Sidford is supported by a Microsoft Research Faculty Fellowship, NSF CAREER Award CCF-1844855, NSF Grant CCF-1955039, a PayPal research award, and a Sloan Research Fellowship.

\bibliographystyle{alpha}
\bibliography{bib.bib}

\appendix

\section{Proof of Properties of Main Building Block}
\label{sec:missing_proofs}

In this section, we give the proof for \Cref{lem:properties_building_block} which we restate below for convenience. 

\PropertyBuildingBlock*

\begin{proof}
We prove the three properties in the lemma statement separately below. 

\smallskip 
\noindent \textbf{Property 1: Non-negativity and boundedness.} 
For any subset $S \subseteq V$, we consider three different cases depending on the relation between $S_A$ and $R$. 

\noindent \textbf{Case 1: $S_A = R$.} In this case, $f_{A,R}(S_A) = 0$ and $\phi_{V,A,R}(S) = 0$, so we have
\begin{align*}
    F(S) = 0 + 0 + \frac{1}{4M |V|} \cdot g(S_B)
\end{align*}
Since $g(S_B) \in [0,M]$ in this case we get $F(S) \in [0, \frac{1}{4|V|}] \in [0,1/4]$.

\noindent \textbf{Case 2: $S_A \subsetneq R$ or $S_A \supsetneq R$.}
In this case, $f_{A,R}(S_A) = 1$ and $|\phi_{V,A,R}(S)| = |S_B| \leq |V|$. Furthermore, $\bone(S_A = R) = 0$.
Thus, 
\[
	F(S) = 1 + \frac{1}{2|V|} \cdot \phi_{V,A,R}(S)
\]
So, in this case, $F(S) \in [0.5, 1.5]$.

\noindent \textbf{Case 3: $S_A$ is neither a subset nor a superset of $R$.} In this case, $f_{A,R}(S_A) = 2$ and $\phi_{V,A,R}(S) = 0$, and therefore $F(S) = 2 \in [0,2]$.  

This completes the proof of Property 1.

\medskip
\noindent \textbf{Property 2: Unique minimizer.} 
An inspection of the cases in the above argument regarding Property 1 shows that for any subset $S$ with $S_A \neq R$, we have $F(S) \geq 0.5$, while when
$S_A = R$, we have $F(S) \leq 0.25$. Therefore, the minimizer $S$ of $F$ must have $S_A = R$.
Furthermore, when $S_A = R$ then $F(S) = \frac{1}{4M |V|} \cdot g(S_B)$ and the function is minimized when $S_B = S_g^*$.
This proves the second property in the lemma statement.

\medskip\noindent \textbf{Property 3: Submodularity.} This is the most interesting part of the proof.  
Let $X, Y \subseteq V$ be two arbitrary subsets of the ground set. Our goal is to prove 
\begin{align} \label{eq:submodularity}
F(X) + F(Y) \geq F(X \cup Y) + F(X \cap Y) .
\end{align}
In the following, we prove \eqref{eq:submodularity} by a case analysis. 
For convenience, define the collection of subsets of $A$ that are either subsets or supersets of $R$ as $\mathcal{H}_{A,R} := \{S \subseteq A: S \subseteq R \text{ or } S \supseteq R\}$.
Note that $R$ lies in this family as well.
We consider three different cases depending on whether or not $X_A$ and $Y_A$ lie in the set family $\mathcal{H}_{A,R}$. 
For notational simplicity, the subscripts in the notations $f_{A,R}$, $\phi_{V,A,R}$ and $\mathcal{H}_{A,R}$ will be dropped throughout the rest of this proof since the sets $V,A,R$ have been fixed and there is no ambiguity.  

\medskip
\def\cH{\mathcal{H}}
\noindent \textbf{(Case 1): $X_A, Y_A \notin \mathcal{H}$.} 
In this case, we have $\phi(X) = \phi(Y) = 0$, $f(X_A) = f(Y_A) = 2$, and $\ind(X_A = R) = \ind(Y_A = R) = 0$.
Thus the $\LHS$ of \eqref{eq:submodularity} is simply $F(X) + F(Y) = f(X_A) + f(Y_A) = 4$. 

Now, note that $(X\cup Y)_A := (X \cup Y) \cap A = X_A \cup Y_A$ and $(X\cap Y)_A := (X\cap Y)\cap A = X_A \cap Y_A$. Therefore, if  $X_A, Y_A \notin \mathcal{H}$, 
then neither $(X\cup Y)_A$ nor $(X\cap Y)_A$ can be $R$. If the former, then both $X_A, Y_A \subseteq R$ implying both are in $\cH$. If the latter, then 
both $X_A, Y_A \supseteq R$ implying both are in $\cH$. Therefore, the 
 $\RHS$ of \eqref{eq:submodularity} doesn't have any ``$g$-terms'', and is
\begin{align*}
    \RHS = f(X_A \cap Y_A) + f(X_A \cup Y_A) + \frac{1}{2|V|} \cdot (\phi(X \cap Y) + \phi(X \cup Y)) .
\end{align*}
Note that if we also have $X_A \cap Y_A, X_A \cup Y_A \notin \mathcal{H}$, then the contribution of $\phi$ to the $\RHS$ would be $0$, and $\LHS \geq \RHS$ follows from the submodularity of $f$ in \Cref{lem:submodularity_f}. 
So we only need to consider the scenarios where $X_A \cap Y_A \in \mathcal{H}$ or $X_A \cup Y_A \in \mathcal{H}$ (or both). 
In any of these scenarios, we have $f(X_A \cap Y_A) + f(X_A \cup Y_A) \leq 3$, since $f(S_A) = 1$ for $S_A \in \cH$. 
Now since $|\phi(S)| \leq |V|$ for any subset $S \subseteq V$, we have $\frac{1}{2|V|} \cdot (\phi(X \cap Y) + \phi(X \cup Y)) \leq 1$. Thus, $\RHS \leq 4$,  and \eqref{eq:submodularity} immediately follows.

\medskip
\noindent\textbf{(Case 2): $X_A, Y_A \in \mathcal{H}$.}
In this case, we need to consider multiple further subcases depending on whether $X_A$ or $Y_A$ coincide with $R$.

\smallskip
\noindent\textbf{Case 2.1: $X_A = Y_A = R$.} In this subcase, $F(S) = f(S_A) + \frac{1}{4M|V|} \cdot g(S_B)$ for all $S \in \{ X, Y, X\cap Y, X \cup Y\}$, so \eqref{eq:submodularity} follows from the submodularity of $f$ and $g$.

\smallskip
\noindent\textbf{Case 2.2: $R \subsetneq X_A, Y_A$.} In this subcase, we have $R \subsetneq X_A \cup Y_A$ and $R \subseteq X_A \cap Y_A$. If it happens that $X_A \cap Y_A = R$, then we have
\begin{align*}
\LHS &= f(X_A) + f(Y_A) + \frac{1}{2|V|} \cdot (\phi(X) + \phi(Y)), \\
\RHS &= f(R) + f(X_A \cup Y_A) + \frac{1}{4M|V|} \cdot g(X_B \cap Y_B) + \frac{1}{2|V|} \cdot \phi(X \cup Y) . 
\end{align*}
Notice that $f(X_A) = f(Y_A) = f(X_A \cup Y_A) = 1$ but $f(R) = 0$. It follows that
\begin{align*}
\LHS - \RHS & = 1 - \frac{1}{2|V|} \cdot (|X_B| + |Y_B|) + \frac{1}{2|V|} \cdot |X_B \cup Y_B| - \frac{1}{4M|V|} \cdot g(X_B \cap Y_B) \\
& = 1 - \frac{1}{2|V|} \cdot |X_B \cap Y_B| - \frac{1}{4M|V|} \cdot g(X_B \cap Y_B) > 0,
\end{align*}
where the last inequality follows because the range of $g$ is within $[0,M]$ by lemma assumption. 
If, on the other hand, that $R \subsetneq X_A \cap Y_A$, then the $\RHS$ of \eqref{eq:submodularity} becomes
\begin{align*}
\RHS = f(X_A \cap Y_A) + f(X_A \cup Y_A) + \frac{1}{2|V|} \cdot (\phi(X \cap Y)  + \phi(X \cup Y)). 
\end{align*}
By a simple counting we have
\begin{align*}
\phi(X) + \phi(Y) = - (|X_B| + |Y_B|) = - (|X_B\cap Y_B| + |X_B\cup Y_B|) = \phi(X \cap Y) + \phi(X \cup Y) . 
\end{align*}
and in this case, $\LHS - \RHS = 0$.

\smallskip
\noindent\textbf{Case 2.3: $X_A, Y_A \subsetneq R$.} The analysis in this subcase is almost identical to \textbf{Case 2.2}. 

\smallskip
\noindent\textbf{Case 2.4: $X_A \subsetneq R \subsetneq Y_A$ or $Y_A \subsetneq R \subsetneq X_A$.} We assume it is the former by symmetry between $X$ and $Y$. Then we have $X_A \cap Y_A = X_A \subsetneq R$ and $X_A \cup Y_A = Y_A \supsetneq R$. From the definition of $F$, it follows that 

\begin{align*}
\LHS - \RHS &= \left(f(X_A) + f(Y_A) - f(X_A \cap Y_A) - f(X_A \cup Y_A)\right) + \\ 
		& ~~~~~~~~~~~~~~~~\frac{1}{2|V|}\cdot \left(\phi(X) + \phi(Y) - \phi(X \cap Y) +-\phi(X \cup Y) \right)
\end{align*}
The first term is $\geq 0$ because of the submodularity of $f$. Furthermore, in this case
\begin{align*}
\phi(X) + \phi(Y) &= \frac{1}{2|V|} \cdot (|X_B| - |Y_B|) \\
\phi(X \cap Y) + \phi(X \cup Y) &= \frac{1}{2|V|} \cdot (|X_B \cap Y_B| - |X_B \cup Y_B|) \leq \frac{1}{2|V|} \cdot (|X_B| - |Y_B|) 
\end{align*}
and thus the second term is also $\geq 0$. 
This proves \eqref{eq:submodularity} in this case.

\smallskip
\noindent\textbf{Case 2.5: $X_A \subsetneq R = Y_A$ or $Y_A \subsetneq R = X_A$.} We assume wlog that it is the former. 
Note that $X_A \cap Y_A = X_A$ and $X_A \cup Y_A = Y_A = R$. 
Therefore,
\begin{align*}
\LHS &= f(X_A) + f(Y_A) + \frac{1}{2|V|} \cdot |X_B| + \frac{1}{4M|V|} \cdot g(Y_B),\\
\RHS &= f(X_A) + f(Y_A) + \frac{1}{2|V|} \cdot |X_B \cap Y_B| + \frac{1}{4M|V|} \cdot g(X_B \cup Y_B).
\end{align*} 

In the above, if $X_B = X_B \cap Y_B$ then it must be that $X_B \subseteq Y_B$. It follows that $Y_B = X_B \cup Y_B$ and we obtain equality in \eqref{eq:submodularity}. On the other hand, if $X_B \neq X_B \cap Y_B$, then $|X_B| \geq |X_B\cap Y_B| + 1$, and so we have 
\begin{align*}
\LHS - \RHS \geq \frac{1}{2|V|} + \frac{1}{4M|V|} \cdot (g(Y_B) - g(X_B \cup Y_B)) \geq \frac{1}{4|V|} > 0,
\end{align*} 
where we used the lemma assumption that the range of $g$ is within $[0,M]$. This again proves \eqref{eq:submodularity}.

\smallskip
\noindent\textbf{Case 2.6: $X_A = R \subsetneq Y_A$ or $Y_A = R \subsetneq X_A$.} 
Assume wlog that it is the former. Then we have
\begin{align*}
\LHS &= f(X_A) + f(Y_A) + 
\frac{1}{4M|V|} \cdot g(X_B) - \frac{1}{2|V|} |Y_B|, \\
\RHS &= f(X_A) + f(Y_A) + \frac{1}{4M|V|} \cdot g(X_B \cap Y_B) - \frac{1}{2|V|} \cdot |X_B \cup Y_B|.
\end{align*} 
The analysis from here is almost identical to that in \textbf{Case 2.5}.

\medskip
\noindent\textbf{(Case 3): $X_A \in \mathcal{H}, Y_A \notin \mathcal{H}$ or $Y_A \in \mathcal{H}, X_A \notin \mathcal{H}$.} We assume wlog that it is the former. 
Note that $f(Y_A) = 2$.
This case is further divide into three subcases below depending on the relation between $X_A$ and $R$.

\smallskip
\noindent\textbf{Case 3.1: $X_A = R$}. And so, $f(X_A) = 0$.
In this subcase, note that $X_A \cap Y_A \subsetneq R$ since $R$ isn't be a subset of $Y_A$, and $R \subsetneq X_A \cup Y_A$. 
So, $f(X_A \cap Y_A) = f(X_A \cup Y_A) = 1$.
Then,
\begin{align*}
\LHS &= f(X_A) + f(Y_A) + \frac{1}{4M|V|} \cdot g(X_B) = 2 + \frac{1}{4M|V|} \cdot g(X_B) \geq 2, \\
\RHS &= f(X_A \cap Y_A) + f(X_A \cup Y_A) + \frac{1}{2|V|} \cdot (|X_B \cap Y_B| - |X_B \cup Y_B|) \leq 2.
\end{align*}
where we used the non-negativity of $g$ in the argument about $\LHS$. In this case, we have established~\eqref{eq:submodularity}.

\smallskip
\noindent\textbf{Case 3.2: $X_A \subsetneq R$}. And so, $f(X_A) = 1$. In this case also, we have $X_A \cap Y_A \subsetneq R$. 
Also note that $X_A \cup Y_A \neq R$ since $Y_A$ is not a subset of $R$. Therefore,
\begin{align*}
\LHS &= f(X_A) + f(Y_A) + \frac{1}{2|V|} \cdot |X_B| = 3 + \frac{1}{2|V|} \cdot |X_B|,\\
\RHS &= f(X_A \cap Y_A) + f(X_A \cup Y_A) + \frac{1}{2|V|} \cdot |X_B \cap Y_B| + \frac{1}{2|V|} \cdot \phi(X \cup Y) \\
& = 1 + f(X_A \cup Y_A) + \frac{1}{2|V|} \cdot |X_B \cap Y_B| + \frac{1}{2|V|} \cdot \phi(X \cup Y) .
\end{align*}
If $X_A \cup Y_A \in \mathcal{H}$, then $f(X_A \cup Y_A) = 1$ and $\phi(X\cup Y) \leq |X_B \cup Y_B| \leq |V|$.
Thus, 
\[
\RHS \leq 2 + \frac{1}{2|V|}|X_B\cup Y_B| + \frac{1}{2|V|}|X_B\cap Y_B| \leq 2.5 + \frac{1}{2|V|}|X_B| < \LHS
\]
Thus, \eqref{eq:submodularity} holds.

If $X_A \cup Y_A \notin \mathcal{H}$, then $f(X_A \cup Y_A) = 2$ and $\phi(X\cup Y) = 0$, and so
$\RHS = 3 + \frac{1}{2|V|} \cdot |X_B \cap Y_B| \leq \LHS$ and thus \eqref{eq:submodularity} holds in this case as well.

\smallskip
\noindent\textbf{Case 3.3: $R \subsetneq X_A$}.
In this case, we have $R \subsetneq X_A \cup Y_A$ and then 
\begin{align*}
\LHS &= f(X_A) + f(Y_A) - \frac{1}{2|V|} \cdot |X_B| = 3 - \frac{1}{2|V|} \cdot |X_B|,\\
\RHS &= f(X_A \cap Y_A) + f(X_A \cup Y_A) + \frac{1}{2|V|} \cdot \phi(X \cap Y) - \frac{1}{2|V|} \cdot |X_B \cup Y_B| \\
& = 1 + f(X_A \cap Y_A) + \frac{1}{2|V|} \cdot \phi(X \cap Y) - \frac{1}{2|V|} \cdot |X_B \cup Y_B|.
\end{align*}
From here one can proceed similarly as in \textbf{Case 3.2} to prove \eqref{eq:submodularity}. 

Combining all the cases above, we established \eqref{eq:submodularity} which implies the submodularity of the function $F$. This completes the proof of the entire lemma. 
\end{proof}

\end{document}